\documentclass[a4paper,10pt]{llncs}
\usepackage[utf8x]{inputenc}
\usepackage{verbatim}
\usepackage{enumerate}
\usepackage{amssymb}
\usepackage{todonotes}
\usepackage{algpseudocode}
\usepackage{graphicx}
\usepackage{textcomp}
\usepackage{algorithm}
\usepackage{caption}

\title{Improved Space efficient linear time algorithms for BFS, DFS and applications\thanks{Some of these results were announced in preliminary form in the proceedings of 22nd International Computing and Combinatorics Conference (COCOON 2016), Springer LNCS volume 9797, pages 119-130 \cite{BanerjeeC016}, and 27th International Symposium on Algorithms and Computation (ISAAC 2016), LIPIcs, volume 64, pages 22:1--22:13 \cite{Chakraborty0S16}. More specifically, this paper contains the linear time algorithms for DFS and its applications, announced in~\cite{BanerjeeC016} and~\cite{Chakraborty0S16}, along with the linear time algorithms for BFS and a few other graph problems, announced in~\cite{BanerjeeC016}.}}
\author{Niranka Banerjee\inst{1}, Sankardeep Chakraborty\inst{1}, Venkatesh Raman\inst{1}, Srinivasa Rao Satti\inst{2}}
\institute{The Institute of Mathematical Sciences, HBNI\\
CIT Campus, Taramani, Chennai 600 113, India\\
\email {nirankab|sankardeep|vraman@imsc.res.in}\\
\and
Seoul National University,\\
1 Gwanak-ro, Gwanak-gu, Seoul, South Korea\\
\email {ssrao@cse.snu.ac.kr}\\
}
\date{}

\begin{document} \maketitle
\begin{abstract}
Research on space efficient graph algorithms, particularly for $st$-connectivity, has a long history including the celebrated polynomial time, $O(\lg n)$ bits \footnote{We use $\lg$ to denote logarithm to the base $2$.} algorithm in undirected graphs by Reingold (JACM 2008), and polynomial time, $n/2^{\Theta(\sqrt{\lg n})}$ bits algorithm in directed graphs by Barnes et al. (SICOMP 1998). Recent works by Asano et al. (ISAAC 2014) and Elmasry et al. (STACS 2015), reconsidered classical fundamental graph algorithms focusing on improving the space complexity. Elmasry et al. gave, among others, an implementation of breadth first search (BFS) in a graph $G$ with $n$ vertices and $m$ edges, taking the optimal $O(m+n)$ time using $O(n)$ bits of space improving the na\"{\i}ve $O(n \lg n)$ bits implementation. Similarly, Asano et al. provided several space efficient implementations for performing depth first search (DFS) in a graph $G$. We continue this line of work focusing on improving the space requirement for 
implementing a few classical graph algorithms.

Our first result is a simple data structure that can maintain any subset $S$ of a universe of $n$ elements using just $n+o(n)$ bits and supports in constant time, apart from the standard insert, delete and membership queries, the operation {\it findany} that finds and returns any element of the set (or outputs that the set is empty). It can also enumerate all elements present currently in the set in no particular order in $O(k+1)$ time where $k$ is the number of elements currently belonging to the set. While this data structure supports a weaker set of operations than that of Elmasry et al. (STACS 2015), it is simple, more space efficient and is sufficient to support a BFS implementation optimally in $O(m+n)$ time using at most $2n+o(n)$ bits. Later, we further improve the space requirement of BFS to at most $n \lg 3+o(n)$
bits albeit with a slight increase in running time to $O(m \lg n f(n))$ time where $f(n)$ is any extremely slow growing function of $n$, and the $o$ term in the space is a function of $f(n)$.

We demonstrate one application of our data structure by developing another data structure using it that can represent a sequence of $n$ non-negative integers $x_1, x_2, \ldots x_n$ using at most $\sum_{i=1}^n x_i + 2n + o(\sum_{i=1}^n x_i+n)$ bits and, in constant time, determine whether the $i$-th element is $0$ or decrement it otherwise. We use this data structure to output, in $O(m+n)$ time and using $O(m+n)$ bits of space, the vertices of 
\begin{itemize}
\item
 a directed acyclic graph in topological sorted order, and
\item
an undirected graph with degeneracy $d$ in degeneracy order.
\end{itemize}
These results improve the space bounds of earlier implementations at least for sparse graphs while maintaining the same linear running time. We also discuss a time-space tradeoff result for finding a minimum weight spanning tree of a weighted (bounded by polynomial in $n$) undirected graph using $n+O(n/f(n))$ bits and $O(m\lg n f(n))$ time, for any function $f(n)$ such that $1 \leq f(n) \leq n$.

For DFS in a graph $G$, we provide an implementation taking $O(m+n)$ time and $O(n \lg m/n)$ bits. This partially answers at least for sparse graphs, a question asked by Asano et al. (ISAAC 2014) whether DFS can be performed in $O(m+n)$ time and using $O(n)$ bits, and also simultaneously improves the DFS result of Elmasry et al. (STACS 2015). Using our DFS algorithm and other careful implementations, we show how one can also test for biconnectivity, $2$-edge connectivity, and find cut vertices and bridges of a given undirected graph within the same time and space bounds; earlier classical linear time algorithms for these problems used $\Omega (n\lg n)$ bits of space.
\end{abstract}

\section{Introduction}

Motivated by the rapid growth of huge data set (``big data''), algorithms that utilize space efficiently are becoming increasingly important than ever before. Another reason for the importance of space efficient algorithms is the proliferation of specialized handheld devices and embedded systems that have a limited supply of memory. 
Even if mobile devices and embedded systems are designed with large supply of memory, it might be useful to restrict the number of write operations specifically for two reasons. One being writing into flash memory is a costly operation in terms of speed and time, and, secondly, it reduces the longevity of such memory. Hence, there is a growing body of work that considers algorithms that do not modify the input and use only a limited amount of work space, and following the recent trend, this paper continues this line of research for fundamental graph algorithms.

\subsection{Model of Computation} 
\label{model1}
We assume that the input graph is given in a read-only memory (and so cannot be modified). If an algorithm must do some outputting, this is done on a separate write-only memory. When something is written to this memory, the information cannot be read or rewritten again. So the input is ``read only'' and the output is ``write only''. In addition to the input and the output media, a limited random-access workspace is available. The data on this workspace is manipulated wordwise as on the standard word RAM, where the machine consists of words of size $\Omega (\lg n)$ bits, and any logical, arithmetic, and bitwise operations involving a constant number of words take a constant amount of time. We count space in terms of the number of bits used by the algorithms in workspace. In other words, storing the input and output is for free, but the input/output space cannot be used by the computation for any other purpose. This model is called the {\it register input model}, and it was introduced by Frederickson \cite{Frederickson87} while studying some problems related to sorting and selection. 

While designing space efficient algorithms in read-only memory model, the specific details of the input graph representation are of great significance as we can neither modify the input nor copy the whole input in workspace. Thus space-efficient algorithms~\cite{Chakraborty0S16,ElmasryHK15,HagerupK16,KammerKL16} assume more powerful form of input representation than what is typically assumed in classical settings. Here we assume that the input graph $G$ is represented using the standard {\it adjacency list along with cross pointers}, i.e., for undirected graphs given a vertex $u$ and the position in its list of a neighbor $v$ of $u$, there is a pointer to the position of $u$ in the list of $v$. In case of directed graphs, for every vertex $u$, we have a list of out-neighbors of $u$ and a list of in-neighbours of $u$. And, finally we augment these two lists for every vertex with cross pointers, i.e., for each $(u,v)\in E$, given $u$ and the position of $v$ in out-neighbors of $u$, there is a pointer to the 
position of $u$ in in-neighbors of $v$. This form of input graph representation was introduced recently in~\cite{ElmasryHK15} and used subsequently in~\cite{BanerjeeC016,HagerupK16,KammerKL16} to design various other space efficient graph algorithms. We note that some of our algorithms will work even with less powerful and the more traditional {\it adjacency list} representation. We specify these details regarding the exact form of input graph representation at the respective sections while describing our algorithms. We use $n$ and $m$ to denote the number of vertices and the number of edges respectively, in the input graph $G$. Throughout the paper, we assume that the input graph is a connected graph, and hence $m \ge n-1$. 


\subsection{Our Results and organization of the paper}
Asano et al.~\cite{AsanoIKKOOSTU14} show that DFS of a directed or undirected graph $G$ on $n$ vertices and $m$ edges can be performed using $n+o(n)$ bits and (an unspecified) polynomial time. Using $2n+o(n)$ bits, they bring down the running time to $O(mn)$ time, and using a larger $O(n)$ bits, the running time of their algorithm is $O(m \lg n)$. In a similar vein,  
\begin{itemize}
\item
we show in Section \ref{BFS} that the vertices of a directed or undirected graph can be listed in BFS order using $n \lg 3+o(n)$
bits and $O(m f(n) \lg n)$ time where $f(n)$ is any (extremely slow-growing) function of $n$ i.e., $\lg^* n$ (the $o$ term in the space is a function of $f(n)$), while the running time can be brought down to the optimal $O(m+n)$ time using $2n+o(n)$ bits.

En route to this algorithm, we develop in Section~\ref{datastructure},
\item
a data structure that maintains a set of elements from a universe of size $n$, say $[1..n]$, using $n+o(n)$ bits to support, apart from the standard insert, search and delete operations, the operation {\it findany} of finding an arbitrary element of the set, and returning its value all in constant time. It can also output all elements of the set in no particular order in $O(k+1)$ time where $k$ is the number of elements currently belonging to the set.

Our structure gives an explicit implementation, albeit for a weaker set of operations than that of Elmasry
et al. [Lemma 2.1,~\cite{ElmasryHK15}] whose space requirement was $cn+o(n)$ bits for
an unspecified constant $c > 2$;\footnote{Since our initial version of the paper, Hagerup and Kammer~\cite{HagerupK16} have independently reported a similar structure with $n + o(n)$ bits for the data structure. See Section~\ref{relwork} for more details.}
furthermore, our structure is simple and is sufficient to 
implement BFS space efficiently, improving by a constant factor of their BFS implementation keeping the running time same.
\end{itemize}
In what follows, in Section \ref{super_linear_bfs}, we improve the space for BFS further at the cost of
slightly increased runtime. We also provide a similar tradeoff for the minimum spanning tree problem in Section \ref{minspan}. In particular, we provide an implementation to find a minimum weight spanning tree in a weighted
undirected graph (with weights bounded by polynomial in $n$) using $n+O(n/f(n))$ bits and $O(m\lg n f(n))$ time, for any function $f(n)$ such that $1 \leq f(n) \leq n$. While this algorithm is similar in spirit to that of Elmasry et al. \cite{ElmasryHK15} which works in $O(m \lg n)$ time using $O(n)$ bits or $O(m+n \lg n)$ time using $O(n \lg (2+(m/n)))$ bits, we work out the constants in the higher order term for space, and improve them slightly though with a slight degradation in time.

\begin{itemize}
\item
Using our data structure, in Section~\ref{topo} 
we develop another data structure to represent a sequence $x_1, x_2, \ldots x_n$ of $n$ non-negative integers using 
$m+2n+ o(m+n)$ bits where $m = \sum_{i=1}^n x_i$. In this, we can determine whether the $i$-th element is $0$ and if not, decrement it, all in constant time.
In contrast, the data structure claimed (without proof) in~\cite{ElmasryHK15} can even change (not just decrement) or access the elements, but in constant {\it amortized} time. However, their structure requires an $O(\lg n)$ limit on the $x_i$ values while we pose no such restriction.
Using this data structure, in the same section, 
\begin{itemize}
\item
we give an implementation of computing a topological sort
of a directed acyclic graph in $O(m+n)$ time and $O(m+n)$ bits of space. We can even detect if the graph is not acyclic within the same time and space bounds. This implementation, contrasts with an earlier bound of $O(m+n)$ time and $O(n \lg \lg n)$ space~\cite{ElmasryHK15}, and is more space efficient for sparse directed graphs (that includes those directed graphs whose underlying undirected graph is planar or has bounded treewidth or degeneracy).
\item
A graph has a degeneracy $d$ if every induced subgraph of the graph has a vertex with degree at most $d$ (for example, planar graphs have degeneracy $5$, and trees have degeneracy $1$). An ordering $v_1, v_2, \ldots v_n$ of the vertices in such a graph is a degenerate order if for any $i$, the $i$-th vertex has degree at most $d$ among vertices $v_{i+1}, v_{i+2}, \ldots v_n$. There are algorithms \cite{Bata,EppsteinLS13} that can find the degeneracy order in $O(m+n)$ time using $O(n)$ words. We show that, given a $d$, we can output the vertices of a $d$-degenerate graph in $O(m+n)$ time using $O(m+n)$ bits of space in the degeneracy order. We can even detect if the graph is $d$-degenerate in the process. As $m$ is $O(nd)$, we have an $O(nd)$ bits algorithm which is more space efficient if $d$ is $o(\lg n)$ (this is the case, for example, in planar graphs or trees).
\end{itemize}
\item For DFS,
Asano et al.~\cite{AsanoIKKOOSTU14} showed that DFS in a directed or undirected graph can be performed in $O(m \lg n)$ time and $O(n)$ bits of space, and Elmasry et al. \cite{ElmasryHK15} improved the time to $O(m\lg \lg n)$ time still using $O(n)$ bits of space. We show the following:
\begin{itemize}
\item In Section \ref{dfs1}, we first show that, we can perform DFS in a directed or undirected graph in linear time using $O(m+n)$ bits. This, for example, improves the runtime of the earlier known results for sparse graphs (where $m$ is $O(n)$) while still using the same asymptotic space. Building on top of this DFS algorithm and other observations, we show how to efficiently compute the {\it chain decomposition} of a connected undirected graph. This lets us perform a variety of applications of DFS (including testing biconnectivity and $2$-edge connectivity, finding cut vertices and edges among others) within the same time and space bound. Our algorithms for these applications improve the space requirement (for sparse graphs) of all the previous algorithms from $\Theta (n\lg n)$ bits to $O(m+n)$ bits, preserving the same linear runtime. 
\item in Section \ref{simp-bicon}, for all the problems mentioned above and dealt in Section \ref{dfs1}, we improve the space even further to $O(n \lg (m/n))$ bits keeping the same $O(m+n)$ running time. The space used by these algorithms, for some ranges of $m$ (say $\Theta (n (\lg \lg n)^c$ for some constant $c$), is even better than that of the recent work by Kammer et al.~\cite{KammerKL16}, that computes cut vertices using $O(n + \min \{ m, n\lg \lg n\})$ bits.
\end{itemize}
\end{itemize}

\subsection{Related Work}
\label{relwork}
In computational complexity theory, the constant work-space model is represented by the complexity class L or DLOGSPACE~\cite{AroraB}. There are several important algorithmic results for this class, most celebrated being Reingold's method for checking reachability between two vertices in an undirected graph \cite{Reingold08}. Barnes et al.~\cite{BarnesBRS98} gave a sub-linear space algorithm for directed graph reachability. Recent work has focused on space requirement in special classes of graphs like planar and H-minor free graphs~\cite{AsanoKNW14,ChakrabortyPTVY14}. Generally these algorithms have very large polynomial running time. In the algorithms literature, where the focus is also on improving time, a huge amount of research has been devoted to memory constrained algorithms, even as early as in the 1980s \cite{MunroP80}. Early work on this focused on the selection problem~\cite{ElmasryJKS14,Frederickson87,MunroP80,MunroR96}, but more recently on computational geometry problems~\cite{AsanoBBKMRS14,AsanoMRW11,BarbaKLSS15,DarwishE14,ElmasryK16} and graph algorithms 
\cite{ElmasryHK15,AsanoIKKOOSTU14,BanerjeeCRRS2015,KammerKL16,Chakraborty0S16,HagerupK16}. Regarding the data structure we develop to support {\it findany} operation, Elmasry et al. [Lemma 2.1,~\cite{ElmasryHK15}] state a data structure (without proof) that supports all the operations i.e., insert, search, delete and findany (they call it $some\_ id$) among others, in constant time. But their data structure takes $O(n)$ bits of space where the constant in the $O$ term is not explicitly stated. Since our initial version of the paper, Hagerup and Kammer \cite{HagerupK16} have independently reported a structure with $n + o(n)$. Though their lower order (the little ``oh'') term $O(n/\lg n)$ is better than ours which is $O(n \lg \lg n/\lg n)$, we believe that our structure is a lot simpler and supports a fewer set of operations, yet sufficient for the space efficient BFS implementation. Furthermore, we do provide a few other applications of our structure also. Recently, Poyias et al.~\cite{ppr} considered the 
problem of compactly representing a rewritable array of bit-strings, and to achieve that they used our findany structure as the main building block in their algorithms. 
Brodal et al.~\cite{BrodalCR96} considered a version of the {\it findany} operation where the goal was to find any element of the set and return its rank (the number of elements smaller than that). For that they gave a non-constant lower bound, though they don't assume that the elements are from a bounded universe. They give a randomized data structure that takes a constant amortized time per operation. However their main objective was to provide time tradeoffs between operations supported by the data structure and they didn't worry about space considerations. We note that this operation and their setup is different from the {\it findany} query we support. 

\subsection{Related Models}
Several models of computation come close to {\it read-only random-access} model-the model we focus on this paper. A single thread common to all of them is that access to the input tape is restricted in some way. In the {\it multi-pass streaming} model \cite{MunroP80} the input is kept in a read-only sequentially-accessible media, and an algorithm tries to optimize on the number of passes it makes over the input. In the {\it semi-streaming} model \cite{Muthukrishnan05}, the elements (or edges if the input is graph) are revealed one by one and extra space allowed to the algorithm is $O(n.polylg(n))$ bits. Observe that, it is not possible to store the whole graph if it is dense. The efficiency of an algorithm in this model is measured by the space it uses, the time it requires to process each edge and the number of passes it makes over the stream. In the {\it in-place} model \cite{BronnimannC06}, one is allowed a constant number of additional variables, but it is possible to rearrange (and sometimes even modify)
 the input values. In the {\it restore} model \cite{ChanMR14}, one is allowed to modify the input but it has to be brought back to its starting configuration afterwards.

\subsection{Preliminaries}
\paragraph{Representing a Vector} We will use the following theorem from \cite{DodisPT10}:
\begin{theorem}\label{nlgc} \cite{DodisPT10}
On a Word RAM, one can represent a vector $A[1..n]$ of elements from a finite alphabet $\Sigma$ using $ n \lg |\Sigma| + O(\lg^2 n) $ 
bits\footnote{The data structure requires $O(\lg n)$ precomputed word constants, thus the second order $O(\lg^2 n)$ bits.}, such that any element of the vector can be read or written in constant time.
\end{theorem}


\paragraph{Rank-Select} We also make use of the following theorem.
\begin{theorem} \cite{Clark96,GuptaHSV07,Munro96}
 \label{staticbit}
We can store a bitstring $O$ of length $n$ with additional $o(n)$ bits such that rank and select operations (defined below) can be supported in $O(1)$ time. Such a structure can also be constructed from the given bitstring in $O(n)$ time.
\end{theorem}
Here the rank and select operations are defined as following:
\begin{itemize}
 \item $rank_a(O,i)$ = number of occurrences of $a\in \{0,1\}$ in $O[1,i]$, for $1\leq i\leq n$;
 \item $select_a(O,i)$ = position in $O$ of the $i$-th occurrence of $a\in \{0,1\}$.
\end{itemize}

\paragraph{Graph theoretic terminology}\label{sec:terms}
Here we collect all the necessary graph theoretic definitions that will be used throughout the paper. In an undirected graph $G$, a cut vertex is a vertex $v$ that when removed (along with its incident edges) from a graph creates more (than what was there before) components in the graph. A (connected) graph with at least three vertices is biconnected (also called $2$-connected or $2$-vertex connected in the literature) if and only if it has no cut vertex. A biconnected component is a maximal biconnected subgraph. These components are attached to each other at cut vertices. Similarly in an undirected graph $G$, a bridge is an edge that when removed (without removing the vertices) from a graph creates more components than previously in the graph. A (connected) graph with at least two vertices is $2$-edge-connected if and only if it has no bridge. A $2$-edge connected component is a maximal $2$-edge connected subgraph. A graph has a degeneracy $d$ if every induced subgraph of the graph has a vertex with degree 
at most $d$. An ordering $v_1, v_2, \ldots v_n$ of the vertices in such a graph is a degenerate order if for any $i$, the $i$-th vertex has degree at most $d$ among vertices $v_{i+1}, v_{i+2}, \ldots v_n$. A topological sort or topological ordering of a directed acyclic graph is a linear ordering of its vertices such that for every directed edge $(u,v) \in E$ from vertex $u$ to vertex $v$, $u$ comes before $v$ in the ordering. A minimum spanning tree (MST) or minimum weight spanning tree is a subset of the edges of a connected, edge-weighted undirected graph that connects all the vertices together, without any cycles and with the minimum possible total edge weight. That is, it is a spanning tree whose sum of edge weights is as small as possible.

\section{Maintaining dictionaries under findany operation}
\label{datastructure}
We consider the data structure problem of maintaining a set $S$ of elements from $\{1, 2, \ldots n\}$ to support the following operations in constant time.
\begin{itemize}
\item {\it insert ($i$)}: Insert element $i$ into the set. 
\item {\it search ($i$)}: Determine whether the element $i$ is in the set.
\item {\it delete ($i$)}: Delete the element $i$ from the set if it exists in the set.
\item {\it findany}: Find any element from the set and return its value. If the set is empty, return a NIL value.
\end{itemize}
Already there exist several solutions for this problem in the data structure literature, but not with the exact time and space bound which we are looking for. In what follows we mention some of these results. It is trivial to support the first three operations in constant time using $n$ bits of a characteristic vector. We could support the {\it findany} operation by keeping track of one of the elements, but once that element is deleted, we need to find another element to answer a subsequent {\it findany} query. This might take $O(n)$ time in the worst case. But this is easy to support in constant time if we have the elements stored in a linked list which takes $O(n \lg n)$ bits. One could also use the dynamic rank-select structure~\cite{HonSS11,RamanRR01} where insert, delete and findany operations take $O(\lg n/\lg \lg n)$ time, and search takes $O(1)$ time. Another approach would be to the classical balanced binary search tree to support these operations. This can be slightly improved further by using the 
van Emde Boas tree. One can improve upon the results for balanced binary search trees and van Emde Boas trees by using the `dynamic range report' structure of Mortensen et al.~\cite{MortensenPP05}, though it still lacks the time and space bound we want here. Our main result in this section is that the {\it findany} operation, along with the other three, can be supported in constant time using $o(n)$ additional bits. We provide a comparison of our findany data structure with previous results in the table below.

\captionof{table}{{\it In the table below, we denote the size of the universe by $n$, and $m$ denotes the number of elements currently present in the data structure. By keeping an additional characteristic bit vector of size $n$, the time required for the search operation can be improved to $O(1)$ time for Balanced BST, vEB and DRR at the expense of extra space. The last column denotes if the update bounds of the corrspeonding data structures are worst case (W) or in expectation (E).}}
\vspace{.3cm}
\begin{tabular}{|l|c|c|c|c|c|c|}
\hline
Data Structure & Insert & Search & Delete & Finadany & Space (in bits) & W/E\\
\hline
  Characteristic Vector  & $O(1)$  & $O(1)$  & $O(1)$ & $O(n)$ & $n$ & W\\
  Balanced Binary Search Tree (BST) & $O(\lg n)$ & $O(\lg n)$ & $O(\lg n)$ & $O(1)$ & $O(m \lg n)$ & W\\
  Dynamic Rank/Select~\cite{HonSS11,RamanRR01} & $O(\frac{\lg n}{\lg \lg n})$
 & $O(1)$ & $O(\frac{\lg n}{\lg \lg n})$ & $O(\frac{\lg n}{\lg \lg n})$ &
$n+o(n)$ & W\\
van Emde Boas tree (vEB)~\cite{Boas75} & $O(\lg \lg n)$ & $O(\lg \lg n)$ & $O(\lg \lg n)$ & $O(\lg \lg n)$ & $O(m \lg n)$ & E \\ 
Dynamic Range Report (DRR)~\cite{MortensenPP05} & $O(\lg \lg n)$ & $O(\lg \lg \lg n)$ & $O(\lg \lg n)$ & $O(\lg \lg \lg n)$ & $O(m \lg n)$ & E\\
Findany (Theorem \ref{maindsbinary} below and~\cite{HagerupK16}) & $O(1)$  & $O(1)$  & $O(1)$ & $O(1)$ & $n+o(n)$ & W\\ \hline
\end{tabular}

\begin{theorem} 
\label{maindsbinary}
A set of elements from a universe of size $n$ can be maintained using $n +o(n)$ bits
to support {\it insert, delete, search and findany} operations in constant time.
We can also enumerate all elements of the set (in no particular order) in $O(k+1)$ time where $k$ is the number of elements in the set. The data structure can be initialized in $O(1)$ time.
\end{theorem}
\begin{proof}
Let $S$ be the characteristic bit vector of the set having $n$ bits.
We follow a two level blocking structure of $S$, as in the case of succinct structures supporting rank and select~\cite{Clark96,Munro96}. However, as $S$ is `dynamic' (in that bit values can change due to insert and delete), we need more auxiliary information. In the discussion below, sometimes we omit floors and ceilings to keep the discussion simple, but they should be clear from the context. 

We divide the bit vector $S$ into $n/\lg^2 n$ blocks of consecutive $\lg^2 n$ bits each, and divide each such block into up to $2\lg n$ sub-blocks of size $\lceil (\lg n)/2 \rceil$ bits each. 
We call a block (or sub-block) non-empty if it contains at least a $1$.
We maintain the non-empty blocks, and the non-empty
sub-blocks within each block in linked lists (not necessarily in order). Within a sub-block, we find the first $1$ or the next $1$ by a table look up. We provide the specific details below.

First, we maintain an array {\it number} indicating the number of $1$s in each block, i.e., $number[i]$ gives the number of $1$s in the $i$-th block of $S$.
It takes $O(n \lg \lg n/\lg^2 n)$ bits as each block can have at most $\lg^2 n$ elements of the given set. 
Then we maintain a queue (say implemented in a space efficient resizable array~\cite{BrodnikCDMS99}) {\it block-queue} 
having the block numbers that have a $1$ bit, and new block numbers are added to the list as and when new blocks get $1$. It can have
at most $n/\lg^2 n$ elements and so has $O(n/\lg^2 n)$ indices taking totally $O(n/\lg n)$ bits.
In addition, every element in {\it block-queue} has a pointer to another queue of sub-block numbers of that block that have an element of $S$. Each such queue has at most $2\lg n$ elements each of size at most $\lg {\lg n}$ bits each (for the sub-block index).
Thus the queue {\it block-queue} along with the queues of sub-block indices takes $O(n \lg \lg n/\lg n)$ bits.
We also maintain an array, {\it block-array}, of size $n/\lg^2 n$ where {\it block-array$[i]$} points to the position of block $i$ in {\it block-queue} if it exists, and is a NIL pointer otherwise and array, {\it sub-block-array}, of size $2n/\lg n$ where {\it sub-block-array$[i]$} points to the position of the subblock $i$ in its block's queue if its block was present in {\it block-queue}, and is a NIL pointer otherwise. So, {\it block-array} takes $n/\lg n$ bits and {\it sub-block-array} takes $2n \lg \lg n/\lg n$ bits.

We also maintain a global table $T$ precomputed that stores for every bitstring of size $\lceil (\lg n)/2 \rceil$, and a position $i$, the position of the first $1$ bit after the $i$-th position. If there is no `next $1$', then the answer stored is $-1$ indicating a NIL value.
The table takes $O(\sqrt n (\lg \lg n)^2)$ bits. This concludes the description of the data structure that takes $n+O(n\lg \lg n/\lg n)$ bits. See Figure $1$
for an illustration.

\begin{figure}[h]
\label{findany_ds_pic}
\begin{center}
 \includegraphics[scale=.8, keepaspectratio=true]{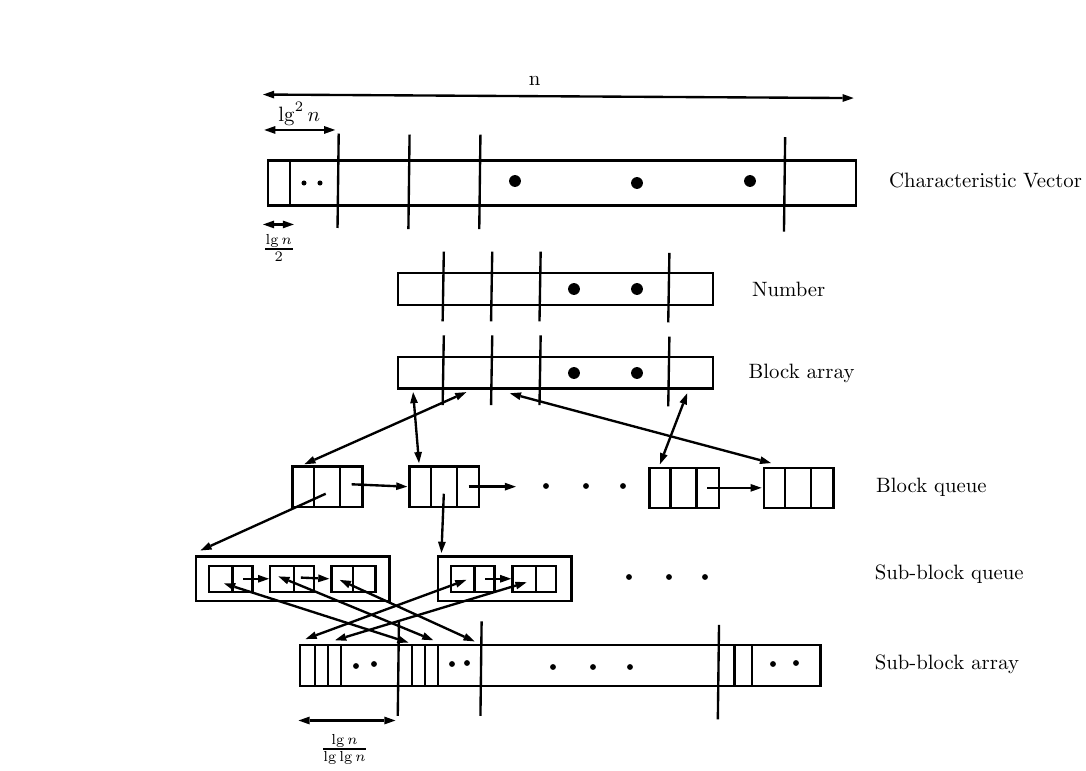}
\end{center}
\caption{An illustration of the inner working details of our findany data structure. The precomputed table is not shown in the diagram.}
\end{figure}

Now we explain how to support each of the required operations. Membership is the easiest: just look at the $i$-th bit of $S$ and answer accordingly. In what follows, when we say the `corresponding bit or pointer', we mean the bit or the pointer corresponding to the block or the sub-block corresponding to an element, which can be determined in constant time from the index of the element. 
To insert an element $i$, first determine from the table $T$, whether there is a $1$ in the corresponding sub-block (before the element is inserted), set the $i$-th bit of $S$ to $1$, and increment the corresponding value in {\it number}. If the corresponding pointer of {\it block-array} was NIL, then insert the block index to {\it block-queue} at the end of the queue, and add the sub-block corresponding to the $i$-th bit into the queue corresponding to the index of the block in {\it block-queue}, and update the corresponding pointers of {\it block-array} and {\it sub-block-array}. If the corresponding bit of {\it block-array} was not NIL (the big block already had an element), and if the sub-block did not have an element before (as determined using $T$), then find the position of the block index in {\it block-queue} from {\it block-array}, and insert the sub-block index into the queue of that block at the end of the queue. Update the corresponding pointer of {\it sub-block-array}.

To support the delete operation, set the $i$-th bit of $S$ to $0$ (if it was already $0$, then there is nothing more to do) and decrement the corresponding number in {\it number}. Determine from the table $T$ if the sub-block of $i$ has a $1$ (after the $i$-th bit has been set to $0$). If not, then find the index of the sub-block from the arrays {\it block-array} and {\it sub-block-array} and delete that index from the block's queue from {\it block-queue}. If the corresponding number in {\it number} remains more than $0$, then there is nothing more to do. If the number becomes $0$, then find the corresponding block index in {\it block-queue} from the array {\it block-array}, and delete that block (along with its queue that will have only one sub-block) from {\it block-queue}. Update the pointers in {\it block-array} and {\it sub-block-array} respectively. As we don't maintain any order in the queues in {\it block-queue}, if we delete an intermediate element from the queue, we can always replace 
that element by the last element in the queue updating the pointers appropriately.

To support the findany operation, we go to the tail of the queue {\it block-queue}, if it is NIL, we report that there is no element in the set, and return the NIL value. Otherwise, go to the block at the tail of {\it block-queue}, and get the first (non-empty) sub-block number from the queue, and find the first element in the sub-block from the table $T$, and return the index of the element.

To enumerate the elements of the set, we traverse the list {\it block-queue} and the queues of each element of {\it block-queue}, and for each sub-block in the queues, we find the next $1$ in constant time using the table $T$ and output the index.

To enable initialization in $O(1)$ time, with each entry of the block (sub-block) queue, we also store the block index corresponding to that entry -- analagous to the ``on-the-fly array initialization'' technique of
[Exercise 2.12 of~\cite{AhoHU74}, Section III.8.1 of~\cite{Mehlhorn84},~\cite{FredrikssonK16}]\footnote{We thank Szymon Grabowski for bringing~\cite{FredrikssonK16} to our attention.}, for example. The (sub) block index stored with the entries in the (sub) block queue act as the ``back pointers'' to the (sub) block-array. Also, instead of storing the precomputed tables to compute the `first 1 bit after the $i$-th position' operation, we can support the operation using $O(1)$ word operations~\cite{fich}.
\end{proof}
We generalize to maintain a collection of more than one disjoint subsets of the given universe to support the insert, delete, membership and findany operations. In this case, insert, delete and findany operations should come with a set index (to be searched, inserted or deleted).
\begin{theorem} 
\label{mainds}
A collection of $c$ disjoint sets that partition the universe of size $n$ can be maintained using 
$n \lg c +o(n)$ bits 
to support {\it insert, delete, search and findany} operations in constant time.
We can also enumerate all elements of any given set (in no particular order) in $O(k+1)$ time where $k$ is the number of elements in the set. The data structure can be initialized in $O(1)$ time.
\end{theorem}
\begin{proof}
The higher order term is for representing the (generalized) characteristic vector $S$ where $S[i]$ is set to the number (index) of the set where 
the element is present. From Theorem~\ref{nlgc}, $S$ can be represented using 
$n \lg c +o(n)$ bits so that the $i$-th value can be retrieved or set in constant time. The rest of the data structures and the algorithms are as in the proof of Theorem~\ref{maindsbinary} (hence the extra $o(n)$ bits), we have a copy of such structures for each of the $c$ sets.
\end{proof}

\section{Breadth First Search}
\label{BFS}


\noindent
We explain how breadth first search (BFS) can be performed in a space efficient manner using the data structure of Theorem \ref{mainds}. Our goal is to output the vertices of the graph in the BFS order. The algorithm is similar in spirit to the $O(n)$ space, $O(m+n)$ algorithm of \cite{ElmasryHK15}, ours is perhaps simpler, we explain the details for completeness.
We start as in the 
textbook BFS by coloring all vertices white. The algorithm grows the search starting at a vertex $s$, making it grey and adding it to a queue. Then the 
algorithm repeatedly removes the first element of the queue, and adds all its white neighbors at the end of the queue (coloring them grey), coloring the element black after removing it from the queue. As the queue can store up to $O(n)$ elements, the space for the queue can be $O(n \lg n)$ bits.
To reduce the space to $O(n)$ bits, we crucially observe the following two properties of BFS:
\begin{itemize}
\item
Elements in the queue are only from two consecutive levels of the BFS tree.
\item
Elements belonging to the same level can be processed in any order, but elements of the lower level must be processed before processing elements of the higher level. 
\end{itemize}

The algorithm maintains four colors: white, grey1, grey2 and black, and represents the vertices with each of these colors as sets $W, S_1, S_2$ and $B$ respectively using the data structure of Theorem~\ref{mainds}. It starts with initializing $S_1$ (grey1) to $s$, $S_2$ and $B$ as empty sets and $W$ to contain all other vertices. Then it processes the elements in each set $S_1$ and $S_2$ switching between the two until both sets are empty. As we process an element from $S_i$, we add its white neighbor to $S_{i+1 \ mod 2}$ and delete it from $S_i$ and add it to $B$.
When $S_1$ and $S_2$ become empty, we scan the $W$ array to find the next white vertex and start a fresh BFS again from that vertex. As insert, delete, membership and findany operations take constant time, and we are maintaining four sets, we have from Theorem ~\ref{mainds}, 
\begin{theorem}
\label{BFS_theorem}
Given a directed or undirected graph $G$, its vertices can be output in a BFS order starting at a vertex using $2n + o(n)$ bits in $O(m+n)$ time.
\end{theorem}

Note that, we don't need to build the findany structure on top of $B$ and $W$ i.e., they can be implemented as plain bitmaps. The findany structures are only required on sets $S_1$ and $S_2$ respectively to efficiently find grey vertices.

In what follows, we slightly deviate from the theme of the paper and show that we can improve the space further for performing BFS if we are willing to settle for more than a linear amount of time. We provide the algorithms below.

\subsection{Improving the space to $n \lg 3 +o(n)$ bits}\label{super_linear_bfs}
There are several ways to implement BFS using just two of the three colors used in the standard BFS \cite{CLRS}, but the space restriction, hence our inability to maintain the standard queue, provides challenges.
For example, suppose we change a grey vertex to white after processing, and use two colors to distinguish the grey vertices in different levels. Then we have the challenge of identifying those white vertices who have been processed (and hence black in the standard BFS) and those yet to be processed (white). For this, we make the following observation.
\begin{itemize}
 \item A vertex $v$ is black if and only if there exists a path consisting of all white vertices from the root $s$ to $v$ (without going through vertices in the sets grey1 or grey2 i.e. without encountering any grey vertices).
\end{itemize}

To check whether such a path exists we can use the reachabilty algorithm of Reingold~\cite{Reingold08} or Barnes et al.~\cite{BarnesBRS98} for an undirected or directed graph respectively on the white vertices (similar idea was used by Asano et al.~\cite{AsanoIKKOOSTU14} in one of their space efficient DFS implementations). Thus we could manage with just $3$ colors, hence using Theorem~\ref{nlgc}, and counting the space required by the reachabilty routine, we can get a BFS implementation using 
$n \lg 3 + o(n)$ 
bits albeit using a large polynomial (as needed by Reingold's or Barnes et al.'s implementation) time. Thus we obtain the following result,
\begin{theorem}
\label{BFS2}
Given a directed or undirected graph $G$, its vertices can be output in a BFS order starting at a vertex in polynomial time
using $n \lg 3 + o(n)$ 
bits of space.
\end{theorem} 

\noindent
Observe that the main bottleneck in the time complexity of the previous algorithm comes from the fact that we are using Reingold's algorithm~\cite{Reingold08} or Barnes et al.'s algorithm~\cite{BarnesBRS98} as subroutine to distinguish 
whether a vertex is black or white. To improve the runtime further (from a large polynomial), we give another $3$ color implementation overloading grey and black vertices, i.e., we use one color to represent grey and 
black vertices. Grey vertices remain grey even after processing. This poses the challenge of separating the grey vertices from the black ones correctly before exploring.

We will have three colors, one for the white unexplored vertices and two colors for those explored including those currently being explored. The two colors indicate the parity of the level (the distance from the starting vertex) of the explored vertices. Thus the starting vertex $s$ is colored $0$ to mark that its distance from $s$ is of even length and every other vertex is colored $2$ to mark them as unexplored (or white). We simply have these values stored in the representation of Theorem~\ref{nlgc} using $n \lg 3 + O(\lg^2 n)$ 
bits and we call this as the color array. The algorithm repeatedly scans this array and in the $i$-th scan, it changes all the $2$ neighbors of $i \ mod \ 2$ to $i+1 \ mod \ 2$. i.e., in one scan of the array, the algorithm changes all the $2$-neighbors of all the $0$ vertices to $1$, and in the next, it changes all the $2$-neighbors of all of the $1$ vertices to $0$. The exploration (of the connected component) stops when in two consecutive scans of the list, no $2$ neighbor is found. Note that for those vertices which would have been colored black in the normal BFS, none of its neighbors will be marked $2$, and so the algorithm will automatically figure them as black. The running time of $O(mn)$ follows because each scan of the list takes $O(m)$ time (to go over neighbors of vertices with one color, some of which could be black) and at most $n+2$ scans of the list are performed as in each scan (except the last two), the color of at least one vertex marked $2$ is changed. Thus we have the following theorem,

\begin{theorem}
\label{bfs2}
Given a directed or undirected graph, its vertices can be output in a BFS order starting at a vertex using $n \lg 3 + O(\lg^2 n)$
bits and in $O(mn)$ time. 
\end{theorem}

The $O(m)$ time for each scan of the previous algorithm is because while looking for vertices labelled $0$ that are supposed to be `grey', we might cross over spurious vertices labelled $0$ that are `black' (in the normal BFS coloring). To improve the runtime further, we maintain two queues $Q_0$ and $Q_1$ each storing up to $n/\lg^2 n$ values to find the grey $0$ and grey $1$ vertices quickly, in addition to 
the color array that stores the values $0$, $1$ or $2$.
We also store two boolean variables, {\it overflow-Q0, overflow-Q1}, initialized to $0$ and to be set to $1$ when more elements are to be added to these queues (but they don't have room). Now the algorithm proceeds in a similar fashion as the previous algorithm except that, along with marking corresponding vertices $0$ or $1$ in the color array, we also insert them into the appropriate queues. i.e. when we expand vertices from $Q_0$ ($Q_1$), we insert their (white) neighbors colored $2$ to $Q_1$ ($Q_0$ respectively) apart from setting their color entries to $1$ ($0$ respectively). Due to the space restriction of these queues, it is not always possible to accomodate all the vertices of some level during the execution of BFS. So, when we run out of space in any of these queues, we continue to make the changes (i.e. $2$ to $1$ or $2$ to $0$) in the color array directly without adding those vertices to the queue, and we also set the corresponding overflow bit. 

Now instead of scanning the color array for vertices labelled $0$ or $1$, we traverse the appropriate queues spending time proportional to the sum of the degree of the vertices in the level. If the overflow bit in the corresponding queue is $0$, then we simply move on to the next queue and continue. When the overflow bit of a queue is set to $1$, then we switch to our previous algorithm and scan the array appropriately 
changing the colors of their white neighbors and adding them to the appropriate queue if possible.
It is easy to see that this method correctly explores all 
the vertices of the graph. The color array uses $n \lg 3 + O(\lg^2 n)$
bits of space. And, for other structures, $Q_0$, $Q_1$ and other variables, we need at most $O(n/\lg n)$ bits. So, overall the space requirement is $n \lg 3 + o(n)$
bits. To analyse the runtime, notice that as long as the overflow bit of a queue is $0$, we spend time proportional the number of neighbors of the vertices in that level, and we spend $O(m)$ time otherwise. When an overflow bit is $1$, then the number of nodes in the level is at least $n/\lg^2 n$ and this can not happen for more than $\lg^2 n$ levels where we spend $O(m)$ time each. Hence, the total runtime is $O(m \lg^2 n)$.

\begin{theorem}
\label{bfs4}
Given a directed or undirected graph, its vertices can be output in a BFS order starting at a vertex using $n \lg 3 + o(n)$
bits of space and in $O(m\lg^2 n)$ time. 
\end{theorem}
\noindent
{\bf Remark:} We can slightly optimize the previous algorithm by observing the fact that a vertex $v$ belongs to an overflowed level then $v$ is expanded twice i.e., first time when the algorithm was expanding vertices from the queue and deleting them. And, secondly, as the overflow bit is set, the algorithm switches to our previous algorithm of Theorem \ref{bfs2} and scan the color array appropriately changing the colors of their white neighbors and adding them to the appropriate queue if possible. We can avoid this double expansion by checking the overflow bit first and if this bit is set, instead of taking vertices out of the corresponding queue and expanding, the algorithm can directly start working with the color array. We can still correctly retrieve all the vertices by checking the same condition. This ensures that all the vertices which belong to an overflowed level won't be expanded twice. 

Note that by making the sizes of the two queues to $O(n/(f(n) \lg n ))$ for any (slow growing) function $f(n)$, the space required for the queues will be $O(n/f(n))$ bits and the running time will be $O(m f(n) \lg n)$.

\begin{theorem}
\label{bfs5}
Given a directed or undirected graph, its vertices can be output in a BFS order starting at a vertex using $n \lg 3 + O(n/f(n))$
bits and in $O(mf(n) \lg n)$ time where $f(n)$ is any extremely slow-growing function of $n$.
\end{theorem}

Since the appearance of the conference version of our paper \cite{BanerjeeC016}, Hagerup et al. \cite{HagerupK16} presented an implementation of BFS taking $n \lg 3  + O(n/ \lg n)$
bits of space and the optimal $O(m+n)$ time, thus  improving the result of our Theorem \ref{bfs5}. But, note that, in our algorithm of Theorem \ref{bfs2}, we improved the space (in the second order) even further than what is needed in their algorithm~\cite{HagerupK16} albeit with degradation in time. We do not know whether we can reduce the space further (to possibly $n + o(n)$ bits) while still maintaining the runtime to $O(m \lg^c n)$ for some constant $c$ or even $O(mn)$. We leave this as an open problem. However, in the next section we provide such an algorithm for MST.

\subsection{Minimum Spanning Tree}
\label{minspan}

In this section, we give a space efficient implementation of the Prim's algorithm~\cite{CLRS} to find a minimum spanning tree. Here we are given a weight function $w: E \rightarrow Z$. We also assume that the weights of non-edges are $\infty$ and that the weights can be represented using $O(\lg n)$ bits. In particular, we show the following,

\begin{theorem}\label{mstproof}
A minimum spanning forest of a given undirected weighted graph, where the weights of any edge can be represented in $O(\lg n)$ bits, can be found using $n + O(n/f(n))$ bits and in $O(m \lg n f(n))$ time, for any function $f(n)$ such that $1 \leq f(n) \leq n$.
\end{theorem}

\begin{proof}
Our algorithm is inspired by the MST algorithm of \cite{ElmasryHK15}, but we work out the constants carefully. Prim's algorithm starts with initializing a set $S$ with a vertex $s$. For every vertex $v$ not in $S$, it finds and maintains 
$d[v] = \min \{ w(v,x): x \in S\}$ and $\pi[v]=x$ where $w(v,x)$ is the minimum among $\{ w(v,y): y \in S\}$. Then it repeatedly deletes the vertex with the smallest $d$ value from $V-S$ adding it to $S$. Then the $d$ values are updated by looking at the neighbors of the newly added vertex.

The space for $d$ values can take up to $O(n \lg n)$ bits.
To reduce the space to $O(n)$ bits, we find and keep, in $O(n)$ time, the set $M$ of the smallest $n/(f(n) \lg n)$ values among the $d$ values of the elements of $V\setminus S$ in a binary heap. This takes $O(n/f(n))$ bits and $O(n)$ time. We maintain the set $S$ in a bit vector taking $n$ bits. We maintain the indices of $M$ in a balanced binary search tree, and each node (index $v$) has a pointer to its position in the heap of the $d$ values, and also stores the index $\pi[v]$. Thus we can think of $M$ as consisting of triples $(v, d[v], \pi[v])$ where $d[v]$ is actually a pointer to $d[v]$ in the heap.
The storage for $M$ takes $O(n/f(n))$ bits. We also find and store the max value of $M$ in a variable $Max$ that also has the vertex label that achieves the maximum.
Now we execute Prim's algorithm by repeatedly deleting elements only from $M$ and updating (decreasing) values in $M$ until $M$ becomes empty. In particular, while updating the values we check if the new value is larger than the variable $Max$. In such cases, we don't do anything. Otherwise, we insert the new value in $M$ and delete the current vertex realizing the maximum value and we proceed further till $M$ becomes empty. Then (for $f(n)\lg n)$ times) we find the next smallest $n/(\lg n f(n))$ values from $V\setminus M\setminus S$ and continue the process. 

Finding the $d$ values of every element in $L= V \setminus S \setminus M$ requires $O(m)$ time (for finding the minimum among all edges incident with vertices in $S$), and finding the smallest $n/f(n) \lg n$ values among them take $O(n)$ time. These steps are repeated $O(f(n) \lg n)$ times resulting in the overall runtime of $O(m f(n) \lg n)$ .
 
In the heap, $n-1$ deletemins and up to $m$ decrease key operations are executed which take $O((m+n)\lg n)$ time by using a binary heap. Note that more sophisticated (for example, Fibonacci heap) implementations are unnecessary as the other operations dominate the running time.
\end{proof}

\section{A Decrement data structure: Applications of Findany dictionary}\label{topo}
In what follows we use our findany data structure of Section~\ref{datastructure} to develop a data structure as below.
\begin{theorem}
\label{decrementds}
Let $x_1, x_2, \ldots x_n$ be a sequence of non-negative integers, and let $m=\sum_{i=1}^n x_i$. Then the sequence can be represented using at most $m+2n + o(m+n)$ bits such that we can determine whether the $i$-th element of the sequence is $0$ and
decrement it otherwise, in constant time.
\end{theorem}
\begin{proof}
Our first attempt is to encode each integer in unary, delimited by a separate bit, to take 
$m+n$ bits. 
A select structure, as in Theorem~\ref{staticbit}, can help us access the corresponding elements. However, decrementing them involves changing this bitstring and so we need a dynamic version of Theorem~\ref{staticbit} that has a $O(\lg n/\lg \lg n)$ runtime for each of the operations~\cite{NavarroS14}.

Now we show how to use our findany structure of Theorem~\ref{mainds} to obtain a linear time algorithm.
We maintain conceptually separate findany structures for each $x_i$ using $x_i+o(x_i)$ bits. 
Each of them stores a subset of $\{1, 2, \ldots x_i\}$ and is initialized to the full universe set.
The total space is $O(m+n) + o(m+n)$ bits. The findany structures of all the integers are concatenated into a single array of bits, separated by a delimiter (say the symbol $2$). And we build a $select$ structure for the delimiter ($2$) in $O(m+n)$ time and using $o(m+n)$ bits of extra space using a generalization of Theorem~\ref{staticbit} for 3 symbols (see, for example \cite{GuptaHSV07})
so that we can identify the findany structure of vertex $i$ in constant time by navigating to the $i$-th delimiter symbol.
Note that though this bitstring will change in the course of an operation, the delimiter symbols aren't modified, and the length of the string doesn't change and so a static $select$ structure suffices. Now decrementing $x_i$ amounts to performing the findany operation on $x_i$'s structure (that actually tests whether $x_i =0)$ and deleting the element if any, that is output from the findany operation.
\end{proof}
 
Using the data structure we just developed, we show the following theorems.
\begin{theorem}
\label{top3}
Given a directed acyclic graph $G$, its vertices can be output in topologically sorted order using $O(m+n)$ time using $m+3n+ o(n+m)$ bits of space. The algorithm can also detect if $G$ is not acyclic.
\end{theorem} 

\begin{proof}
A standard algorithm repeatedly outputs a vertex with indegree zero and deletes that vertex along with its outgoing edges, 
until there are no more vertices. To implement this, we maintain first the set $Z$ of indegree $0$ vertices in the datastructure of Theorem~\ref{maindsbinary} to support findany operation in constant time. This takes $O(m+n)$ time and $n+o(n)$ bits. We also represent the indegree sequence of the vertices using the data structure of Theorem~\ref{decrementds}.
The algorithm repeatedly finds any element from $Z$, outputs and deletes it from $Z$. Then it
decrements the indegree of its out-neighbors, and includes them in $Z$ if any of them has become $0$ (that can be determined by another findany structure) in the process. If $Z$ becomes empty even before all elements are output (that can be checked using a counter or a bit vector), then at some intermediate stage of the algorithm, we did not encounter a vertex with indegree zero which means that the graph is not acyclic. 
\end{proof}

An undirected graph is $d$-degenerate if every induced subgraph of the graph has a vertex with degree at most $d$. For example, a graph with degree at most $d$ is $d$-degenerate. A planar graph is $4$-degenerate as every planar graph has a vertex with degree at most $4$. The degeneracy order of a $d$-degenerate graph is an ordering $v_1, v_2, \ldots v_n$ of the vertices such that $v_i$ has degree at most $d$ among $v_{i+1}, v_{i+2}, \ldots v_n$. We show the following using our data structure developed in this section.
\begin{theorem}
\label{degenerate}
Given a $d$-degenerate graph $G$, its vertices can be output in $d$-degenerate order using 
$m+3n + o(m+n)$ bits and $O(m+n)$ time. The algorithm can also detect if the given graph is not $d$-degenerate.
\end{theorem}

\begin{proof}
As in the topological sort algorithm, we maintain the set $Z$ of vertices whose degree in the entire graph is at most $d$, 
using our findany data structure. This takes $O(m+n)$ time and $n+o(n)$ bits. Then we represent the degree sequence of the 
vertices using the data structure of Theorem~\ref{decrementds} except we subtract $d$ from each of them. I.e. 
$x_i = max \{0, d_i - d\}$ where $d_i$ is the degree of the $i$-th vertex. The algorithm repeatedly finds any element from $Z$, outputs and deletes it from $Z$. Then it decrements the degree of its neighbors, and includes them in $Z$ if any of them has become $0$. If $Z$ becomes empty even before all elements are output (that can be checked using a counter or a bit vector), then at some intermediate stage of the algorithm, we did not encounter a vertex with degree less than $d$,  which means that the graph is not $d$-degenerate. 
\end{proof}

Note that, for all the algorithms discussed in the last two sections i.e., Section \ref{BFS} and \ref{topo} respectively, we can assume that the input graph is represented as the standard {\it adjacency list}~\cite{CLRS} instead of the more powerful adjacency list along with cross pointers representation.


\section{DFS and its applications using $O(m+n)$ bits} \label{dfs1}
The classical and standard implementation of DFS using a stack and color array takes $O(m+n)$ time and $O(n \lg n)$ bits of space. Improving on this, recently Elmasry et al.~\cite{ElmasryHK15} showed the following,
\begin{theorem}\label{general_dfs}
A DFS traversal of a directed or undirected graph $G$ with $n$ vertices and $m$ edges can be performed using $O(n \lg \lg n)$ bits of space and $O(m+n)$ time. 
\end{theorem}
In this section, we start by improving upon the results of
Theorem~\ref{general_dfs} of Elmasry et al.~\cite{ElmasryHK15} and Theorem $4$ of Asano et al.~\cite{AsanoIKKOOSTU14} by showing an $O(n)$-bit DFS traversal method for sparse graphs that runs in linear time. Using this DFS as backbone, we provide a space efficient implementation for computing several other useful properties of an undirected graph.

\subsection{DFS}
In what follows we describe how to perform DFS in $O(n+m)$ time using $O(n+m)$ bits of space. Note that, this is better (in terms of time) than both the previous solutions for sparse graphs (when $m=o(n \lg \lg n)$) with same space bounds. The class of sparse graphs includes a large class of graphs including planar graphs, bounded genus, bounded treewidth, bounded degree graphs, and H-minor-free graphs. These are also the majority of the graph classes which arise in practice, and as our algorithm is simple with no heavy data structures used, we believe that an implementation of our algorithm will work very fast in practice for these graphs. 

Recall that, our input graphs $G=(V,E)$ are represented using the standard adjacency array along with cross pointers. We describe our algorithm for directed graphs, and mention the changes required for undirected graphs. Central to our algorithm is an encoding of the out-degrees of the vertices in unary. Let $V =\{1, 2, \cdots, n\}$ be the vertex set. The unary degree sequence encoding $O$ of the directed graph $G$ has $n$ $0$s to represent the $n$ vertices and each $0$ is followed by a number of $1$s equal to the out-degree of that vertex. Moreover, if $d$ is the degree of vertex $v_i$, then $d$ $1$s following the $i$-th $0$ in the $O$ array corresponds to $d$ out-neighbors of $v_i$ (or equivalently the edges from $v_i$ to the $d$ out-neighbors of $v_i$) in the same order as in the out-adjacency array of $v_i$. Clearly $O$ uses $n+m$ bits and can be obtained from the out-neighbors of each vertex in $O(m+n)$ time. We use another bit string $E$ of the same length where every bit is initialized to $0$. The 
array $E$ will be used to mark the tree edges of the DFS as we build the DFS tree, and will be used to backtrack when the DFS has finished exploring a vertex. The bits in $E$ are in one-to-one correspondence with bits in $O$. If $(v_i, v_j)$ is an edge in the DFS tree where $v_i$ is the parent of $v_j$, and suppose $k$ is the index of the edge $(v_i, v_j)$ in $O$, then the corrsponding location in the $E$ array is marked as $1$ during DFS. Thus once DFS finishes traversing the whole graph, the number of ones in the $E$ array is exactly the number of 
tree edges. We also store another array, say $C$, having entries from \{{\it white, gray, black\}} with the usual meaning i.e., each vertex $v$ remains white until it is visited, is colored gray when DFS visits $v$ for the first time, and is colored black when its out-adjacency array has been checked completely. We can represent $C$ using Lemma \ref{nlgc} in $n \lg 3 + o(n)$ bits so that individual entries can be accessed or modified in constant time. The bitvector $O$ is represented using the static rank-select data structure of Theorem~\ref{staticbit} that uses additional $o(m+n)$ bits. So overall we need $2m+(\lg 3+2)n + o(m+n)$ bits to represent the arrays $O, E$ and $C$.

Suppose $v_j$ is a child of $v_i$ in the final DFS tree. We
can think of the DFS procedure as performing the following two steps repeatedly until all the vertices are explored. 
First step takes place when DFS discovers a vertex $v_j$ for the first time, and as a result $v_j$'s color changes to gray from white. We call this phase as forward step. When DFS completes exploring $v_j$ i.e. the subtree rooted at $v_j$ in the DFS tree, it performs two tasks subsequently. First, it backtracks to its parent $v_i$, and then finds in $v_i$'s list the next white neighbor to explore. The latter part is almost similar to the forward step described before. We call the first part alone as backtrack step. In what follows, we describe how to implement each step in detail.  

We start our DFS with the starting vertex, say $r$, changing its color to gray in the color array $C$. 
Then, as in the usual DFS algorithm, we scan the out-adjacency list of $r$, and find the first white neighbor, say $v$, to make it gray. When the edge $(r, v)$ is added to the DFS tree, we mark the position corresponding to the edge $(r, v)$ in $E$ to~$1$\footnote{Note that to implement this step, we only require the $select_0$ operation in Theorem~\ref{staticbit}.}. We continue the process with the new vertex making it gray until we encounter a vertex $w$ that has no white out-neighbors. At this point, we will color the vertex $w$ black, and we need to backtrack.

To find the vertex to backtrack, we do the following. We go to $w$'s in-neighbor list to find a gray vertex which is its parent. For each gray vertex $t$ in $w$'s in-neighbor list, we follow the cross pointers to reach $w$ in $t$'s out-adjacency list and check its corresponding entry $(t,w)$ in $E$ array (using select operation to find $w$ after $t$-th $0$). 
Observe that, among all these gray in-neighbors of $w$, only one edge out of them to $w$ will be marked in $E$
as this is the edge that DFS traversed while going in the forward direction to $w$. So once we find an in-neighbor $t$ such that the position corresponding to $(t,w)$ in $E$ is marked and $t$ is gray, we know that $w$'s parent is $t$ in the DFS tree. 
Also the cross pointer puts us in the position of $w$ in $t$'s out-neighbor list, and we start from that position to find the next white vertex to explore DFS. So the only extra computation we do is to spend time proportional to the degree of each black vertex (to find its parent to backtrack) and so overall there is an extra overhead of $O(m)$ time. The navigation we do to determine the tree edges are on $O$ which is a static array, and so from Theorem~\ref{staticbit}, all these operations can be performed in constant time. Thus we have

\begin{theorem}
\label{directed1}
A DFS traversal of a directed graph $G$ can be performed in $O(n+m)$ time using $(2m+(\lg 3+2)n) + o(m+n)$ bits.
\end{theorem}

For undirected graphs, first observe that the unary degree sequence encoding $O$ takes $2m+n$ bits as each edge appeares twice. As $E$ also takes $2m+n$ bits, overall we require $(4m+(\lg 3+2)n) + o(m+n)$ bits of space. As for DFS, observe that the forward step, as defined before, can be implemented in exact same manner. It is crucial to mention one subtle point that, while marking an edge $(v_i,v_j)$, we don't mark its other entry i.e. $(v_j,v_i)$. So when DFS finishes, {\it for tree edges exactly one of the two entries will be marked one in $E$ array}. Backtracking step is now little easier as we don't have to switch between two lists. We essentially follow the same steps in the adjacency array to check for a vertex $t$ in $w$'s array such that $t$ is gray and the corresponding entry for the edge $(t,w)$ is marked in $E$. Once found, we start with the next white vertex. Hence, 
\begin{theorem}
\label{undirected1}
A DFS traversal of an undirected graph $G$ can be performed in $O(n+m)$ time using $(4m+(\lg 3+2)n) + o(m+n)$ bits.
\end{theorem}
We can decrease the space slightly by observing that, we are not really using the third color {\it black}. More specifically, we can continue to keep a vertex gray even after its subtree has been explored. As we only explore white vertices always and never expand gray or black, the correctness follows immediately. This gives us the following.
\begin{theorem}
\label{mixed}
A DFS traversal of a directed graph $G$ can be performed in $O(n+m)$ time using $(2m+3n) + o(m+n)$ bits. For undirected graphs, the space required is $(4m+3n) + o(m+n)$ bits.
\end{theorem} 

\subsection{Applications of DFS}
One of the classical applications of DFS is to determine, in a connected undirected graph, all the cut vertices and bridges which are defined as, respectively, the vertices and edges whose removal results in a disconnected graph. Since the early days of designing graph algorithms, researchers have developed several approaches to test biconnectivity and $2$-edge connectivity, find cut vertices and bridges of a given undirected graph. Most of these methods use depth-first search as the backbone to design the main algorithm. For biconnectivity and $2$-edge connectivity, the classical algorithm due to Tarjan~\cite{Tarjan72,Tarjan74} computes the so-called ``low-point" values (which are defined in terms of a DFS-tree) for every vertex $v$, and checks some conditions using that to determine cut vertices, bridges of $G$ and check whether $G$ is $2$-edge connected or biconnected. Brandes~\cite{Brandes02} and Gabow~\cite{Gabow00} gave considerably simpler algorithms for testing biconnectivity by using simple path-
generating rules instead of low-points; they call these algorithms path-based. All of these algorithms take $O(m+n)$ time and $O(n)$ words of space. Another algorithm due to Schmidt \cite{Schmidt13} is based on chain decomposition of graphs to determine biconnectivity and $2$-edge connectivity. Implementing this algorithm takes $O(m+n)$ time and $O(m)$ words of space. In what follows, we present a space efficient implementation for Schmidt's algorithm based on the DFS algorithm we designed in the previous section. We summarize our result in the theorem below.
\begin{theorem}
\label{dfsapps}
Given a connected undirected graph $G$, in $O(m+n)$ time and using $O(m+n)$ bits of space we can determine whether $G$ is $2$-vertex (and/or edge) connected. 
If not, in the same amount of time and space, we can compute all the bridges and cut vertices of the graph. 
\end{theorem}

Schmidt~\cite{Schmidt2010c} introduced a decomposition of the input graph that partitions the edge set of the graph into cycles and paths, called chains, and used this to design an algorithm to find cut vertices and biconnected components \cite{Schmidt13} and also to test 3-connectivity~\cite{Schmidt2010c} among others. We briefly recall Schimdt's algorithm and its main ingredient of {\it chain decomposition}. The algorithm first performs
 a depth first search on $G$. Let $r$ be the root of the DFS tree $T$. DFS assigns an index to every vertex $v$ i.e. the time vertex $v$ is discovered for the first time (discovery time) during DFS. Call it depth-first-index ($DFI(v)$). Imagine that the
the back edges are directed away from $r$
and the tree edges are directed towards $r$. The algorithm
decomposes the graph into a set of paths and cycles called
chains as follows. See Figure $2$ for an illustration.

\begin{figure}[h]
\begin{center}
 \includegraphics[scale=.7, keepaspectratio=true]{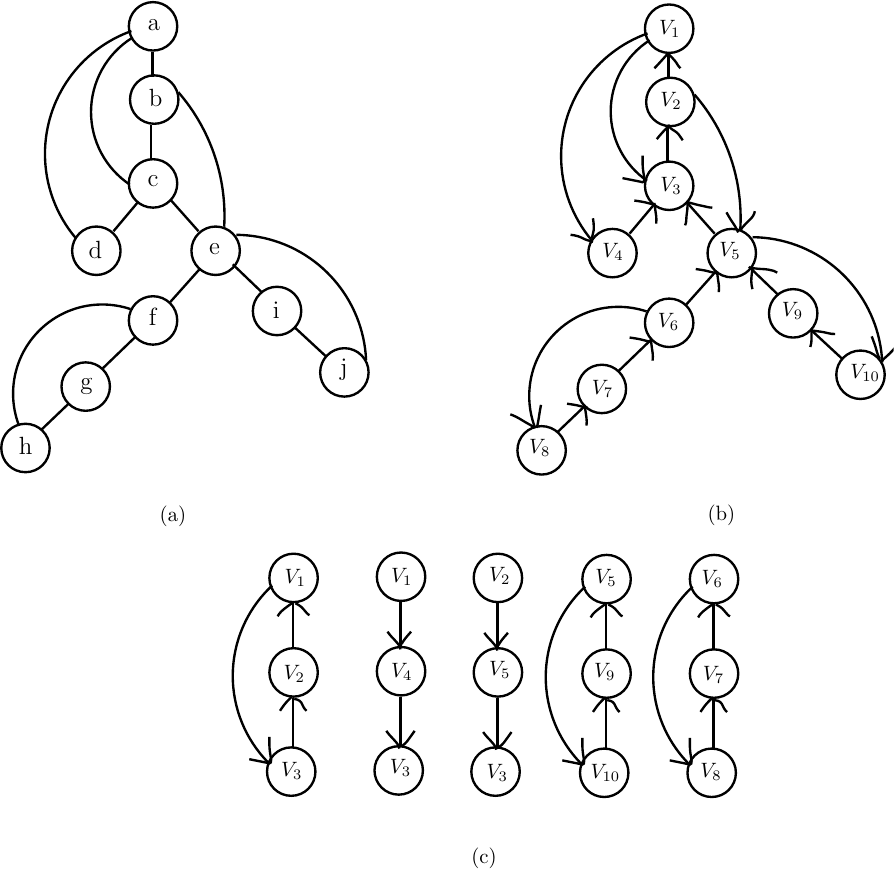}
\end{center}
\caption{Illustration of Chain Decomposition. (a) An input graph $G$. (b) A DFS traversal of $G$ and the resulting edge-orientation along with DFIs. (c) A chain decomposition $D$ of $G$. The chains $D_2$ and $D_3$ are paths and rest of them are cycles. The edge $(V_5,V_6)$ is bridge as it is not contained in any chain. $V_5$ and $V_6$ are cut vertices.}
\end{figure}

First we mark all the vertices as unvisited. Then we visit every vertex starting at $r$ in increasing order of DFI, and do the following.
For every back edge $e$ that originates at $v$, we traverse a directed cycle or a path.
This begins with $v$ and the back edge $e$ and proceeds along the tree towards the root and stops at 
the first visited vertex or the root. During this step, we mark every encountered vertex visited. This forms the first chain. Then we proceed with the next back edge at $v$, if any, or move towards the next $v$ in increasing DFI order and continue the process.
Let $D$ be the collection of all such cycles and paths. Notice
that, the cardinality of this set is exactly the same as the number of
back edges in the DFS tree as each back edge contributes to one cycle or
a path. Also as initially every vertex is unvisited, the first chain would be a cycle as it would end in the starting vertex. 
Schmidt proved the following theorem.

\begin{theorem} \cite{Schmidt13} 
\label{2ec} 
Let $D$ be a chain decomposition of a connected graph $G(V,E)$. Then $G$ is 2-edge-connected if and if the chains in $D$ partition $E$. Also, $G$ is 2-vertex-connected if and if $\delta(G) \geq 2$ (where $\delta(G)$ denotes the minimum degree of $G$) and $D_1$ is the only cycle in the set $D$ where $D_1$ is the  first chain in the decomposition. An edge $e$ in $G$ is bridge if and if $e$ is not contained in any chain in $D$. A vertex $v$ in $G$ is a cut vertex if and if $v$ is the first vertex of a cycle in $ D \setminus D_1$. 
\end{theorem}

The algorithm (the tests in Theorems~\ref{2ec}) can be implemented easily in $O(m+n)$ time using $O(m+n)$ words as we can store the DFIs
and entire chain decomposition $D$. 
To reduce the space to $O(m+n)$ bits, we first perform a depth first
search of the graph $G$ (as mentioned in Theorem \ref{undirected1}) and recall that at the
end of the DFS procedure, we have the color array $C$ with all colors black and the array $E$ which encodes the DFS tree. Here for a tree edge $(i,j)$ where $i$ is closer to the root, the position corresponding to the edge $(i,j)$ is 
marked $1$ and that corresponding to $(j,i)$ is marked $0$, and the backedges are marked $0$. To implement the chain decomposition, we do not have space to store the chains or the DFS indices. To handle the latter (DFI), we (re)run DFS and then use Schmidt's algorithm along with DFS in an interleaved way. Towards the end, we recolor all the vertices to white. To handle the former, we use two more arrays, one to mark the vertices visited in the chain decomposition, called {\it visited} and another array $M$, to mark the edges visited during the chain decomposition. The array $M$ has size $(n+2m)$ bits, and it has the same initial 
structure as $E$ i.e. $0$'s separated by $1$'s where $0$'s
denote edges and $1$'s denote vertices. 
The details of forming the chain decomposition and finding all cut vertices and bridges 
using these arrays $O$ (original outdegree encoding), $E$ (the DFS tree), $C$ (color array), $visited$ and $M$ (to mark edges) are explained below.

\begin{proof} {\it of Theorem \ref{dfsapps}.} We explain here the details of forming the chain decomposition and finding all cut vertices and bridges. 
We start at the root vertex $r$, and using the array $E$, find the first `back edge' (non-tree edge) $(r,x)$ to $r$.
This can be found by going to the $r$-th $0$ in $O$ and then to the corresponding position in $E$ that represents the vertex $r$ (note that $E$ has a lot more zeroes, and so we should get to the corresponding $0$ of $r$ in $E$ by first getting to the corresponding position in $O$). If $O$ has $1$s after the corresponding $0$, then we look for the first $0$ after the corresponding position in $E$ to find the back edge (as all the tree edges are marked $1$).
We mark $r$ and $x$ visited (if they were unvisited before) and mark both copies of the edge $(r,x)$ (unlike what we do in the forward step of DFS) using the cross pointer in $M$.
Now to obtain the chain, we need to follow the tree edges from $x$. We use the `backtracking' procedure we used earlier 
for DFS. We look for an (the only) edge marked $1$ in $E$ out of the edges incident on $x$ by scanning the adjacency list, and that gives the parent $y$ of $x$ (Here is where we use the fact we only mark one copy of the edge as we explore the DFS tree.). 

We continue after marking $y$ visited, and the edge $(x,y)$ (both copies) in $M$ until we reach $r$ or a visited vertex when we complete the chain. Now we continue from where we left of in $r$'s neighborhood to look for the next back edge and continue this process. Once we are done with back edges incident on $r$, we need to proceed to the next vertex in DFS order.
As we have not stored the Depth First Indices, we essentially (re)run the DFS using the color array $C$. For this, we flush out the color array to make every vertex white again. 
Note that we don't make any changes to array $E$ and $O$ respectively. As this DFS procedure is deterministic, it will follow exactly the same sequence of paths like before, ultimately leading to the same DFS tree structure, and note that, this structure is already saved in array $E$. 

Clearly, the amount of space taken is $O(m+n)$ bits. 
To analyze the runtime, note that, we first perform a DFS traversal which takes linear time. At the second step, we basically perform one more round of DFS. 
As a visited node is never explored (using the {\it visited} array), the overall runtime is $O(m+n)$. Edge connectivity (Theorem~\ref{2ec}) can easily be checked using the array $M$ once we have the chain decomposition. The bridges are the edges marked $0$ in the array $M$. Cut vertices can be obtained and listed out if and and when we reach the starting vertex while forming a chain, except at the first chain (if exists).
This completes the proof.
\end{proof}

Combining all the main results from this section, we summarize our results in the following theorem below,
\begin{theorem}\label{dfsapps1}
A DFS traversal of an undirected or directed graph $G$ can be performed in $O(m+n)$ time using $O(m + n)$ bits. In the same amount of time and space, given a connected undirected graph $G$, we can perform a chain decomposition of $G$, and using that we can determine whether $G$ is $2$-vertex (and/or edge) connected. If not, in the same amount of time and space, we can compute all the bridges and cut vertices of $G$.
\end{theorem}


In what follows, we show in Section \ref{simp-bicon} how to improve the space bounds of Theorem \ref{dfsapps1} keeping the same running time by applying different bookeeping technique but essentially using the same algorithm itself.  

\section{DFS and applications using $O(n \lg (m/n))$ bits}
\label{simp-bicon}
As mentioned previously, one can easily implement the tests in Theorem~\ref{2ec} in $O(m+n)$ time using $O(m)$ words, by storing the DFIs and the entire chain decomposition, $D$. Theorem \ref{dfsapps1} shows how to perform the tests using $O(m+n)$ bits and $O(m+n)$ time. The central idea there is to maintain the DFS tree using $O(m+n)$ bits using an unary encoding of the degree sequence of the graph. And later, build on top of it another $O(m+n)$ bits structure to perform chain decompositions and the other tests of Schimdt's algorithm. We first show how the space for the DFS tree representation can be improved to $O(n \lg m/n)$ bits. Also note that, all the algorithms from the last section assumes that the input graph must be respresented as adjacency array with cross pointers. Our algorithms in this section also get rid of this assumption. Here we only assume that the input graph is represented as a standard adjacency array. We start by proving the following useful lemma. 

\begin{lemma}\label{lem:adjlist-pointers}
Given the adjacency array representation of an undirected graph $G$ on $n$ vertices with $m$ edges, 
using $O(m)$ time, one can construct an auxiliary structure of size $O(n \lg (m/n))$ bits that can 
store a ``pointer'' into an arbitrary position within the adjacency array of each vertex. Also, updating 
any of these pointers (within the adjacency array) takes $O(1)$ time.
\end{lemma}

\begin{proof}
We first scan the adjacency array of each vertex and construct a bitvector $B$ 
as follows: starting with an empty bitvector $B$, for $1 \le i \le n$, if $d_i$ is 
the length of the adjacency array of vertex $v_i$ (i.e., its degree), then we append the string 
$0^{\lceil{\lg d_i}\rceil -1}1$ to $B$. The length of $B$ is $\sum_{i=1}^n \lceil{\lg d_i}\rceil$, 
which is bounded by $O(n \lg (m/n))$. We construct auxiliary structures to support $select$ queries 
on $B$ in constant time, using Theorem~\ref{staticbit}.
We now construct another bitvector $P$ of the same size as $B$, which stores pointers into the adjacency arrays of each vertex. The 
pointer into the adjacency array of vertex $v_i$ is stored using the $\lceil{\lg d_i}\rceil$ bits in $P$ from 
position $select(i-1,B)+1$ to position $select(i,B)$, where $select(0,B)$ is defined to be $0$. Now, using 
select operations on $B$ and using constant time word-level read/write operations, one can access and/or 
modify these pointers in constant time.
\end{proof}

Given that we can maintain such pointers into the lists of every vertex, the following lemma shows that, within the same time and space bounds, we can actually maintain the DFS tree of a given graph $G$. Details are provided below.

\label{sec:biconn-parentprts}
\begin{lemma}\label{lem:parent-pointers}
Given a graph $G$ with $n$ vertices and $m$ edges, in the adjacency array representation in the read-only memory model,
the representation of a DFS tree can be stored using $O(n \lg (m/n))$ additional bits, which can be constructed
on the fly during the DFS algorithm.
\end{lemma}

\begin{proof}
We use the representation of Lemma~\ref{lem:adjlist-pointers} to store {\it parent} pointers into the adjacency array of 
each vertex. In particular, whenever the DFS outputs an edge $(u,v)$, where $u$ is the parent of $v$, we 
scan the adjacency array of $v$ to find $u$ and store a pointer to that position (within the adjacency array of $v$).
The additional time for scanning the adjacency arrays adds upto $O(m)$ which would be subsumed by the running 
time of the DFS algorithm.
\end{proof}

We call the representation of the DFS tree of Lemma~\ref{lem:parent-pointers} as the {\em parent pointer representation}. Now given Lemma \ref{lem:adjlist-pointers} and \ref{lem:parent-pointers}, we can simulate the DFS algorithm of Theorem \ref{undirected1} to obtain an $O(n \lg (m/n))$ bits and $O(m+n)$ time DFS implementation. The proof of Theorem~\ref{dfsapps} then uses another $O(m+n)$ bits to construct the chain decomposition of $G$ and perform the tests as mentioned in Theorem~\ref{2ec}, and we show here how even the space for the construction of a chain decomposition and performing the tests can be improved. We summarize our results in the following theorem below:

\begin{theorem}\label{thm:biconn-parentptrs}
A DFS traversal of an undirected or directed graph $G$ can be performed in $O(m+n)$ time using $O(n \lg (m/n))$ bits of space. In the same amount of time and space, given a connected undirected graph $G$, we can perform a chain decomposition of $G$, and using that we can determine whether $G$ is $2$-vertex (and/or edge) connected. If not, in the same amount of time and space, we can compute and report all the bridges and cut vertices
of $G$.
\end{theorem}

\begin{proof}
Given Lemma~\ref{lem:adjlist-pointers} and~\ref{lem:parent-pointers}, it is easy to verify that we can simulate the DFS algorithms of Theorem~\ref{directed1} and~\ref{undirected1} to obtain an $O(n \lg (m/n))$ bits and $O(m+n)$ time DFS implementation. In what follows we use this DFS algorithm to perform the tests in Theorem~\ref{2ec}. With the help of the parent pointer representation, we can visit every vertex, starting at the root $r$ of the DFS tree, in increasing order of DFI, and enumerate (or traverse through) all the non-tree (back) edges of the graph as required in Schimdt's algorithm as follows: for each node $v$ in DFI order, and for each node $u$ in its adjacency list, we check if $u$ is a parent of $v$. If so, then $(u,v)$ is a tree edge, else it is a back edge. We maintain a bit vector {\it visited} of size $n$, corresponding to the $n$ vertices, initialized to all zeros meaning all the vertices are unvisited at the beginning. We use {\it visited} array to mark vertices visited 
during the chain decomposition. When a new back edge is visited for the first time, the algorithm traverses the path starting with the back edge followed by 
a sequence of tree edges (towards the root) untill it encounters a marked vertex, and also marks all the vertices on this path.
By checking whether the vertices are marked or not, we can also distinguish whether an edge is encountered for the first time 
or has already been processed. Note that this procedure constructs the chains on the fly.

To check whether an edge is a bridge or not, we first note that only the tree edges can be bridges 
(back edges always form a cycle along with some tree edges). Also, from Theorem~\ref{2ec},
it follows that any (tree) edge that is not covered in the chain decomposition algorithm is a bridge.
Thus, to report these, we maintain a bitvector $M$ of length $n$, corresponding to the $n$ vertices, 
initialized to all zeros. Whenever a tree edge $(u,v)$ is traversed during the chain decomposition
algorithm, if $v$ is the child of $u$, then we mark the child node $v$ in the bit vector $M$. After 
reporting all the chains, we scan the bitvector $M$ to find all unmarked vertices $v$ and output
the edges $(u,v)$, where $u$ is the parent of $v$, as bridges. If there are no bridges found in this process, then $G$ is $2$-edge connected. To check whether a vertex is a cut vertex (using the characterization in 
Theorem~\ref{2ec}), we keep track the starting vertex of the current chain (except for the first chain, which is a cycle), 
that is being traversed, and report that vertex as a cut vertex if the current chain is a cycle. If there are no cut vertices found in this process then $G$ is $2$-vertex connected. Otherwise, we keep one more array of size
$n$ bits to mark which vertices are cut vertices. This completes the proof.

\end{proof}

Note that, all of our algorithms in this section don't use cross pointers and hence we can just assume that the input graph $G$ is represented via {\it adjacency array} i.e., given a vertex $v$ and an integer $k$, we can access the $k$-th neighbor of vertex $v$ in constant time. We remark that this input representation was also used in \cite{Chakraborty0S16,KammerKL16} recently to design various other space efficient graph algorithms. 

\section{Conclusions}
\label{concl}
We have provided several implementations of BFS focusing on optimizing space without much degradation in time. In particular with $2n+o(n)$ bits we get an optimal linear time algorithm whereas squeezing space further gives an algorithm with running time $O(m f(n) \lg n)$ where $f(n)$ can be any (extremely slow-growing) function of $n$. One can immediately 
obtain similar time-space tradeoffs for natural applications of BFS including testing whether a graph is bipartite or to obtain all connected components of a graph. To achieve this, we also provide a simple and space efficient implementation of the {\it findany} data structure. This data structure supports in constant time, apart from the standard insert, delete and membership queries, the operations findany and enumerate. Very recently, Poyias et al.~\cite{ppr} considered the problem of compactly representing a rewritable array of bit-strings, and for this, they heavily used our findany data structure. It would be interesting to find other such applications for our data structure. Using the findany data structure, we also provide a space efficient implementation of the {\it decrement} data structure which supports in constant time the decrement and check if any element is zero operations. We use this decrement data structure to design space efficient algorithms for performing topological sort in a directed 
acyclic graph and computing degeneracy ordering of a given undirected graph. For the MST problem, we could reduce the space further to $n+o(n)$ bits. It is an interesting question whether we can perform BFS using $n+o(n)$ bits with a runtime of $O(m \lg^c n)$ for some constant $c$ or even $O(mn)$. 

For DFS, we provide an $O(m+n)$ time and $O(n \lg (m/n))$ space DFS traversal method. For a large class of graphs including planar, bounded degree and bounded treewidth graphs, this gives an $O(n)$ bits and $O(m+n)$ time DFS algorithm, partially answering an open question in~\cite{AsanoIKKOOSTU14}, and improving the result of~\cite{ElmasryHK15}. Within the same time and space bound, we also show how to test biconnectivity and $2$-edge connectivity, obtain cut vertices and bridges, and compute a chain decomposition of a given undirected graph $G$. It is a challenging and interesting open problem whether DFS and all these applications can be performed using $O(m+n)$ time and $O(n)$ bits. See ~\cite{Chakraborty0S16,KammerKL16} for some recent algorithms using $O(n)$ bits for most of the problems considered in this paper and others, albeit with a slightly more than linear running time.\\

{\bf Acknowledgement} We thank Saket Saurabh for suggesting the question that led to results in Section 5. \\

\bibliographystyle{plain}

\end{document}